\documentclass{LMCS}

\def\dOi{10(3:15)2014}
\lmcsheading%
{\dOi}
{1--24}
{}
{}
{Mar.~18, 2014}
{Sep.~10, 2014}
{}

\ACMCCS{[{\bf Theory of computation}]: Computational complexity and cryptography---Problems, reductions and completeness}

\keywords{random strings, truth-table degrees, strong reducibilities,
  minimal pair} 

\usepackage{amsfonts,amsmath,amsthm,amssymb,url,ifpdf,xspace}

\def\phi{\varphi}
\def\res{\!\!\upharpoonright\!}     

\def\+{\oplus}
\def\*{\odot}

\newcommand{\0}{\mathbf{0}}

\newcommand{\ba}{\begin{eqnarray*}}
\newcommand{\ea}{\end{eqnarray*}}
\newcommand{\bes}{\begin{enumerate}\topsep=1.5mm \itemsep=-1mm}
\newcommand{\ee}{\end{enumerate}}

\newcommand{\ov}{\overline}

\newcommand\leT{\leq_{\rm T}}

\newcommand{\lett}{\leq_{\rm {tt}}}

\newcommand{\es}{\emptyset}

\newcommand{\ce}{c.e.\ }

\newcommand{\RKUO}{R_1}
\newcommand{\RKUT}{R_2}
\newcommand{\RKUj}{R_j}
\newcommand{\KUj}{K_j}

\newcommand{\ei}{{\langle e,i\rangle}}

\newcommand{\eizero}{{\langle e_0,i_0\rangle}}

\newcommand{\ns}{{\langle n,s\rangle}}

\newcommand{\meizerox}{{m_{\eizero,x}}}

\newcommand{\eip}{{\langle e',i'\rangle}}

\def\ba{{\bf a}}

\def\uh{\upharpoonright}
\newcommand{\<}{\langle}
\renewcommand{\>}{\rangle}
\newcommand{\concat}{\symbol{94}}

\begin{document}

\title[Random strings and tt-degrees of Turing complete C.E.\ sets]%
{Random strings and truth-table degrees of Turing complete C.E.\ sets}

\author[M.~Cai]{Mingzhong Cai\rsuper a}
\address{{\lsuper a}Department of Mathematics \\
  Dartmouth College \\
  Hanover, NH 03755, USA}
\email{Mingzhong.Cai@dartmouth.edu}
\thanks{{\lsuper a}The first author's research was partially supported
  by NSF grant DMS-1266214.}


\author[R.~G.~Downey]{Rodney G. Downey\rsuper b}
\address{{\lsuper b}Department of Mathematics, Statistics, and Operations
  Research \\
  Victoria University of Wellington \\
  P.O. Box 600 \\
  Wellington, NEW ZEALAND}
\email{rod.downey@msor.vuw.ac.nz}
\thanks{{\lsuper b}The second author's research was partially supported by a
  Marsden grant.}


\author[R.~Epstein]{Rachel Epstein\rsuper c}
\address{{\lsuper c}Department of Mathematics and Statistics\\
  Swarthmore College \\
  500 College Ave \\
  Swarthmore, PA 19081, USA}
\email{rachel.epstein@swarthmore.edu}
\thanks{{\lsuper c}The third author's research was partially supported by an
AMS-Simons Foundation Travel Grant.}


\author[S.~Lempp]{Steffen Lempp\rsuper d}
\address{{\lsuper{d,e}}Department of Mathematics \\
  University of Wisconsin \\
  Madison, WI 53706-1388, USA}
\email{\{lempp,jmiller\}@math.wisc.edu}

\thanks{{\lsuper d}The fourth author's research was partially
  supported by AMS-Simons Foundation Collaboration Grant 209087.}

\author[J.~S.~Miller]{Joseph S. Miller\rsuper e}
\thanks{{\lsuper e}The last author's research was partially supported
  by NSF grant DMS-1001847.}

\amsclass{Primary: 68Q30, 03D25, 03D30; Secondary: 68Q15, 03D32}

\begin{abstract}
We investigate the truth-table degrees of (co-)c.e.\ sets, in particular,
sets of random strings. It is known that the set of random
strings with respect to any universal prefix-free machine is Turing complete,
but that truth-table completeness depends on the choice of universal machine.
We show that for such sets of
random strings, any finite set of their truth-table degrees do not meet to
the degree~$\mathbf0$, even within the \ce truth-table degrees, but when
taking the meet over all such truth-table degrees, the infinite meet is
indeed~$\mathbf0$. The latter result proves a conjecture of Allender,
Friedman and Gasarch. We also show that there are two Turing
complete c.e.\ sets whose truth-table degrees form a minimal pair.
\end{abstract}

\maketitle\vfill


\section{Introduction}
\label{sec_intro}

Recent work in theoretical computer science has established a link between
computational complexity classes and the languages efficiently reducible to
sets of random strings. Intuitively, a random (finite, binary) string is one
that does not have a shorter description than itself. We use \emph{Kolmogorov
complexity} to formalize this intuition, but this leaves us with choices:
There are two common types of Kolmogorov complexity (plain and prefix-free),
and within each type, the set of random strings depends on the choice of
universal machine. See Section~\ref{background} for more detail. Irrespective
of these choices, however, the set of random strings can be shown to ``speed
up'' computation.

\begin{thm}[Buhrman, Fortnow, Kouck\'y and Loff~\cite{BFKL:10}; Allender,
Buhrman, Kouck\'y, van Melkebeek and Ronneburger~\cite{ABKMR:06}; Allender,
Buhrman and Kouck\'y~\cite{ABK:06}]\label{thm:low} Let~$R$ be the set of all
random strings for either plain or prefix-free complexity.
\begin{itemize}
	\item $\textup{BPP}\subseteq \textup{P}^R_\text{tt}$.
	\item $\textup{PSPACE}\subseteq \textup{P}^R$.
	\item $\textup{NEXP}\subseteq \textup{NP}^R$.
\end{itemize}
\end{thm}

\noindent So, for example, a language in PSPACE can be recognized by a
polynomial-time machine with access to~$R$.

It is also possible to give upper bounds for what can be efficiently
reducible to the set of random strings.

\begin{thm}[Allender, Friedman and Gasarch~\cite{AFG}]\label{thm:AFG}\
\begin{itemize}
\item $\Delta^0_1\;\cap\;\bigcap_{U}\textup{P}_\text{tt}^{R_{K_U}}
    \subseteq\textup{PSPACE}$.
\item $\Delta^0_1\;\cap\;\bigcap_{U}\textup{NP}^{R_{K_U}}
    \subseteq\textup{EXPSPACE}$.
\end{itemize}
Here~$U$ ranges over universal prefix-free machines,~$K_U$ is prefix-free
complexity as determined by~$U$, and~$R_{K_U}$ is the corresponding set of
random strings.
\end{thm}

Taking the intersection over all universal prefix-free machines has the
effect of ``factoring out'' the choice of machine. Why is this necessary?
First note that~$R_{K_U}$ is not computable (hence not in PSPACE) but it is
efficiently reducible to itself. This example is unsatisfying, of course,
because we are already explicitly restricting to computable (i.e.,
$\Delta^0_1$) languages. For a better example, note that is possible to build
a universal prefix-free machine~$U$ for which there is a computable set $A\in
\textup{P}_\text{tt}^{R_{K_U}}$ that is not in EXPSPACE.\footnote{For plain
complexity, this follows from~\cite[Theorem~12]{ABK:06}. The authors point
out that the same proof works in the case of prefix-free complexity. It
remains open if for \emph{every} universal prefix-free machine~$U$, there is
a computable set
$A\in\textup{P}_\text{tt}^{R_{K_U}}\smallsetminus\text{EXPSPACE}$, even
though the corresponding fact holds for plain complexity. See the discussion
after \cite[Theorem~16]{ABK:06}.}

There are two other ways that Theorem~\ref{thm:AFG} is restricted. For one,
it is only stated for prefix-free complexity; Allender, Friedman and
Gasarch~\cite{AFG} conjecture that it holds for plain complexity as well.
More important for our purposes is the explicit restriction to computable
languages. Allender et al.\ conjecture that this restriction is redundant.

\begin{conj}[Allender, Friedman and Gasarch~\cite{AFG}]\label{conj}
If $A\in \bigcap_{U}\textup{NP}^{R_{K_U}}$, then~$A$ is computable.
(Therefore, $\Delta^0_1\;\cap$ can be removed from both parts of
Theorem~\ref{thm:AFG}.)
\end{conj}

\noindent We prove this conjecture and study related questions.

Our approach is purely computability-theoretic. Any set in $\textup{NP}^{R}$
is truth-table reducible to~$R$, so we study tt-reduction to sets of random
strings, i.e., sets of the form~$R_{K_U}$ for different choices of the
prefix-free universal machine~$U$. We show in Theorem~\ref{thmfinite} that
every finite collection of sets of random strings can tt-compute some
noncomputable computably enumerable set. Note that sets of random strings are
Turing complete (and co-c.e.), so it is reasonable to ask if
Theorem~\ref{thmfinite} is a special case of a more general restriction on
the tt-degrees of Turing complete \ce sets. It is not;
Theorem~\ref{minpairthm} shows that there is a minimal pair of Turing
complete \ce sets within the tt-degrees. Finally, in
Theorem~\ref{knight-bishop} we prove that there is \emph{no} noncomputable
set that is tt-reducible to \emph{every} set of random strings. This verifies
Conjecture~\ref{conj}.

Putting Theorems~\ref{thm:low} and~\ref{thm:AFG} together with
Conjecture~\ref{conj}, we obtain:
\begin{itemize}
\item $\textup{BPP}\subseteq\bigcap_{U}\textup{P}_\text{tt}^{R_{K_U}}
    \subseteq\textup{PSPACE}$
\item $\textup{NEXP}\subseteq\bigcap_{U}\textup{NP}^{R_{K_U}}
    \subseteq\textup{EXPSPACE}$
\end{itemize}
In each case,~$U$ ranges over universal prefix-free machines.
Allender~\cite{allender} conjectures that the lower bounds are tight, i.e.,
that $\textup{BPP} = \bigcap_{U}\textup{P}_\text{tt}^{R_{K_U}}$ and
$\textup{NEXP} = \bigcap_{U}\textup{NP}^{R_{K_U}}$, but this is still
\emph{very much} an open question.

\subsection{Definitions and background}\label{background}

The Kolmogorov complexity of a finite string \mbox{$\sigma\in 2^{<\omega}$}
is a measure of how difficult it is to describe~$\sigma$. Let~$M\colon
2^{<\omega}\to 2^{<\omega}$ be partial computable function (we call such a
function a \emph{machine}). The plain complexity of~$\sigma$ with respect
to~$M$ is
\[
C_M(\sigma)=(\mu n)(\exists \tau)[U=M(\tau)=\sigma\ \ \& \ \ |\tau|=n].
\]
This depends on the choice of~$M$, but it is straightforward to check that
there is a \emph{universal} machine~$U$ such that~$C_U$ is optimal for such
machines, up to an additive constant. Plain Kolmogorov complexity~$C$ is
defined to be~$C_U$ for a fixed universal machine~$U$. Note that for any two
universal machines~$U$ and~$V$, $C_U(\sigma)\leq C_V(\sigma)+c$ for some
constant~$c$ depending on~$U$ and~$V$.

We define \emph{prefix-free Kolmogorov complexity} in a similar manner. We
say that a machine~$M\colon 2^{<\omega}\to 2^{<\omega}$ is \emph{prefix-free}
if whenever~$\sigma$ and~$\tau$ are two distinct strings contained in the
domain of~$M$, then neither is a prefix of the other (i.e.,
$\sigma\mid\tau$).  A \emph{universal} prefix-free machine is one that can
simulate all other prefix-free machines. Prefix-free complexity with respect
to a universal prefix-free machine~$U$ is written~$K_U(\sigma)$, and is
defined in the same way as~$C_U(\sigma)$.  Similarly,~$K(\sigma)$ is
$K_U(\sigma)$ for some fixed universal prefix-free machine~$U$.  As before,
the choice of~$U$ can make at most a finite difference.

As a notational convention, we use~$[s]$ after a term to mean the state of
that term after~$s$ stages.  For instance, $K(\sigma)[s]=K_U(\sigma)[s]$ is
the shortest length of any~$\tau$ such that $U_s(\tau)=\sigma$.

While plain complexity~$C$ at first seems like the most natural way to define
complexity, it lacks some properties that we would expect a complexity
measure to have.  For example, it is not true that there is a constant~$c$
such that $C(\sigma\tau)\leq C(\sigma)+C(\tau)+c$.  Thus, to describe the
concatenation of~$\sigma$ and~$\tau$, we cannot simply provide descriptions
of both strings along with some finite code for concatenation, as we would
expect that we could.  Prefix-free complexity~$K$ satisfies more properties
that we would desire a complexity measure to satisfy.  For instance, there
does exist a constant~$c$ such that $K(\sigma\tau)\leq K(\sigma)+K(\tau)+c$.
(For more information, see~\cite[p.~83]{Nies},
or~\cite[p.~121]{Downey-Hirschfeldt}.)

Intuitively, for a string to be random, it should have no description shorter
than its own length.  This leads to the following definitions.  Let~$U$ be a
universal prefix-free machine.  We define the set of random strings with
respect to~$U$ by
$$
R_{K_U}=\{\sigma \mid K_U(\sigma)\geq |\sigma|\}.
$$
For a fixed prefix-free universal machine~$U$, we let $R_K=R_{K_U}$.
Similarly, we can define~$R_{C_V}$ and~$R_C$ using plain complexity and a
standard universal machine~$V$.

Note that while the choice of machine makes only a small difference in
complexity, the sets $R_{K_{U_1}}$ and $R_{K_{U_2}}$ could potentially be
quite different.  Thus, we cannot talk about a given string~$\sigma$ being
``random'' without specifying a machine.  In this paper, we look at how
different the sets of random strings with respect to different machines can
be.

It is known that both~$R_C$ and~$R_K$ are Turing complete, regardless of the
choices of universal machines; see Li and
Vit\'anyi~\cite[Exer\-cise~2.7.7]{LiVitanyi} for details.  (The exercise
states that the set $\{\langle x,y\rangle \mid C(x)\leq y\}$ is Turing
complete, but the proof uses only~$\ov R_C$.  It is not difficult to extend
the proof to the prefix-free case, as the machine built to compress strings
in the exercise can easily be made to be prefix-free.)  In fact, by the same
argument,~$R_C$ and~$R_K$ are always bounded Turing-complete (or,
bT-complete, for short). That is, we can computably find a bound for the use
function in the computation that reduces the halting set to~$R_C$
or~$R_K$.\footnote{Note that bT-reducibility is also known as weak
truth-table reducibility, or wtt-reducibility.} Thus, comparing the Turing
degrees or bounded Turing degrees of two sets of random strings will not help
to differentiate them. So we turn instead to truth-table reducibility, the
next finer reducibility.

Truth-table reducibility is a strengthening of Turing reducibility and
bT-reduci\-bi\-lity.  For an arbitrary Turing functional~$\Phi$, there is no
computable way to know for which oracles~$A$ and which input~$m$ $\Phi^A(m)$
converges.  There is also no way to know how much of the oracle is needed to
perform a given computation.  For a truth-table reduction (tt-reduction),
this information can be computably known.  There are two standard ways of
defining a truth-table reduction.  One way is as a {\em total Turing}
reduction.  That is, a Turing functional~$\Psi$ is a tt-reduction if for
every oracle $A\in 2^\omega$,~$\Psi^A$ is a total function.  The other way to
define a truth-table reduction is using truth tables.  Each
tt-reduction~$\Psi$ is given by two computable functions,~$f$ and~$g$.  The
value~$f(m)$ gives a number that can be thought of as a bound for the use of
the computation~$\Psi^A(m)$ for any oracle~$A$. The value~$g(m)$ gives a code
that tells us, for each $\sigma\in 2^{f(m)}$, what the value of
$\Psi^\sigma(m)$ is.  Thus, we can think of~$f$ and~$g$ as defining a table
whose rows consist of every string~$\sigma$ of length~$f(m)$ and the
corresponding output $\Psi^\sigma(m)$.  If $A=\Psi^B$ for a
tt-reduction~$\Psi$, we write $A\lett B$.

We can effectively list the tt-reductions by also including some reductions
that are not total and are therefore not tt-reductions.  We let
$\{\Psi_i\}_{i\in \omega}$ be a listing of the tt-reductions in the following
way.  Let $i=\langle i_f, i_g\rangle$, where $(x,y) \mapsto \langle
x,y\rangle$ is the standard Cantor pairing function from $\omega\times
\omega$ to~$\omega$. Let~$\Psi_i$ be the reduction given by the functions~$f$
and~$g$ as above, where $f=\phi_{i_f}$ and $g=\phi_{i_g}$ (and where
$\{\phi_e\}_{e\in\omega}$ is the standard listing of the partial computable
functions).  We say that the truth table for~$\Psi_i(m)$ has been defined
after~$s$ steps if $\phi_{i_f}(m)[s]$ and $\phi_{i_g}(m)[s]$ both converge
and $\phi_{i_g}(m)$ codes the values for the rows of the table given by
$\phi_{i_f}(m)$.  If either function does not converge or if the functions
cannot be interpreted as giving a truth table, then the truth table is
undefined.

It is not hard to show that $A\lett B$ if and only if $A\leT B$ via a Turing
reduction that runs in a computably bounded time.  Truth-table reductions are
thus closely connected to computer science.

In our work below, we build on ideas from two beautiful theorems on the
tt-degrees of sets of random strings. The first is about the set of strings
random with respect to plain complexity:

\begin{thm}[Kummer~\cite{Kummer}]\label{KummerThm}
$R_C$ is truth-table complete.
\end{thm}

\noindent Kummer's theorem does not depend on the choice of universal machine
used to define~$R_C$.  Thus, every \ce set is contained in $\bigcap_U\{A \ |
\ A\lett R_{C_U}\}$, where the intersection is taken over every universal
machine.  This does not hold for the prefix-free case, by the second result:

\begin{thm}[Muchnik~\cite{Muchnik}]\label{MuchnikThm}
There exists a prefix-free universal machine~$U$ such that~$R_{K_U}$ is not
truth-table complete.
\end{thm}

\noindent In fact, in Theorem~\ref{knight-bishop} we show that $\bigcap_U\{A
\ | \ A\lett R_{K_U}\}=\Delta^0_1$.  That is, the only sets tt-reducible to
every~$R_{K_U}$ are the computable sets. Our proof relies heavily on ideas
that were introduced in Muchnik's proof, in particular, the idea of playing a
game to force the value of a truth-table reduction. It is worth noting that
the proof of Theorem~\ref{thm:AFG} was also inspired by Muchnik's proof.

Kummer's proof served as the basis of our proof of Theorem~\ref{thmfinite},
where we show that for any finite collection of~$R_{K_U}$'s, there is a
noncomputable \ce set tt-reducible to each~$R_{K_U}$. Essentially, we try to
transfer Kummer's coding method to the prefix-free case. While we cannot
code~$\emptyset'$, we do find that we can code some noncomputable \ce set.


\section{\texorpdfstring{There is no tt-minimal pair of $R_{K_U}$'s}%
{There is no tt-minimal pairs of U-random sets}}\label{sec_nominpair}

We first state and prove the following theorem for two~$R_{K_U}$'s and then
generalize it in Theorem~\ref{thmfinite} to the case of finitely
many~$R_{K_U}$'s.

\begin{thm}\label{theorem1}
For any prefix-free universal machines~$U_1$ and~$U_2$, there is a
noncomputable \ce set~$A$ such that $A\lett R_{K_{U_1}}$ and $A\lett
R_{K_{U_2}}$.
\end{thm}

For notational simplicity, let $K_j=K_{U_j}$ and $R_j=R_{K_{U_j}}$ for
$j=1,2$.

We use~$U_1$ as the universal prefix-free machine that gives us the
prefix-free complexity~$K(\sigma)$, so $K(\sigma)=K_1(\sigma)$.  We use the
usual correspondence between finite strings and natural numbers to define
$\KUj(n)$ for $j=1, 2$.  That is,~$\sigma$ is the string corresponding to~$n$
if~$1\sigma$ is the binary representation of $n+1$.
By Chaitin's Counting Theorem~\cite{Chaitin}, there is a constant~$c$ such
that $|\{\sigma\in 2^n \mid \KUj(\sigma)<|\sigma|\}|< 2^{n-K(n)+c}$, for each
$j=1, 2$; that is, the number of length~$n$ strings in~$\ov{R_j}$ is bounded
by $2^{n-K(n)+c}$.

Let~$g(n)$ be the computable function, defined by Solovay in~\cite{Solovay},
with the property that $K(n)\leq g(n)$ for all $n\in \omega$ and such that
$g(n)=K(n)$ on an infinite set.  There is, however, no infinite \ce set on
which $g(n)=K(n)$ (\cite{Solovay}, see~\cite[p.~132]{Gacs}).
We will construct an infinite set~$A$ that is truth-table reducible to
both~$\RKUO$ and~$\RKUT$, and such that if~$A$ is computable, then there is
an infinite \ce set on which $g(n)=K(n)$.  Thus,~$A$ is not computable,
showing that~$\RKUO$ and~$\RKUT$ do not form a minimal pair.  Also note that
there is a constant~$b$ such that $K(n)\leq b+2\log n$, so we may assume that
$g(n)\leq b+2\log n$.

We will simultaneously construct two prefix-free machines~$M_1$ and~$M_2$.
Using the Recursion Theorem, we may assume that we know in advance that the
coding constants of machine~$M_j$ with respect to~$U_j$ are less than some
value~$d$ for each $j=1,2$; that is, $K_{U_j}(\sigma)<K_{M_j}(\sigma)+d$ for
all~$\sigma$ and each $j=1,2$.
%
The purpose of the machines will be to compress strings, which will
force~$U_1$ and~$U_2$ to compress strings, which will allow us to code
information into~$\RKUO$ and~$\RKUT$.

As we have said, our proof is inspired by Kummer's proof of
Theorem~\ref{KummerThm}.  The main idea is that we know by Chaitin's Counting
Theorem that the number of nonrandom strings of length~$n$ (with respect to
either~$U_1$ or~$U_2$) is less than $2^{n-g(n)+c}$ for all~$n$ such that
$g(n)=K(n)$.  We can divide the set
of natural numbers less than
$2^{n-g(n)+c}$ into $2^{c+d}$ many regions of size $2^{n-g(n)-d}$, for
all~$n$.  We know there is some maximal such region such that the size of the
set of nonrandom strings of length~$n$ lies in that region, for infinitely
many~$n$ with $K(n)=g(n)$.
We can code information into~$U_j$ by waiting until $K(n)[s]=g(n)$ and
choosing $2^{n-g(n)-d}$ strings that will be compressed if~$K(n)$ drops
below~$g(n)$.  For almost all~$n$ in the maximal region, these strings will
only be compressed if $K(n)<g(n)$, because otherwise we would contradict the
maximality of the region.  Thus, we are compressing strings to code
information about which elements~$n$ satisfy $K(n)<g(n)$.

We construct our machines by enumerating KC (Kraft-Chaitin) sets.  A KC set
is a \ce set of pairs $\{\langle d_i, \sigma_i\rangle\}_{i\in\omega}$ from
$\omega \times 2^{<\omega}$ such that the weight
$$
\sum_{i \in \omega} 2^{-d_i}
$$
of the set is at most~$1$.  By the KC Theorem, also known as the Machine
Existence Theorem, a KC set determines a prefix-free machine~$M$ such that
$M(\tau_i)=\sigma_i$ with $|\tau_i|=d_i$ for all $i\in \omega$.  Thus, any
universal prefix-free machine must also compress~$\sigma_i$ to length~$d_i$
plus a coding constant.  (The KC Theorem is due to Levin~\cite{Levin},
Schnorr~\cite{Schnorr}, and Chaitin~\cite{Chaitin}.  See
also~\cite[p.~125]{Downey-Hirschfeldt}.)

To build our KC sets, we first build sets~$E^j_n$ such that~$E^j_n$ contains
strings of length~$n$, for $j=1,2$. We then enumerate $\langle n-d, \sigma
\rangle$ into a KC set for each $\sigma\in \bigcup_{n\in\omega}E^j_n$ to
define machine~$M_j$. We will construct~$E^j_n$ so that if $g(n)=K(n)$,
then~$E_n^j$ will be empty; and otherwise (i.e., if $g(n)>K(n)$) we have
$|E^j_n|\leq 2^{n-g(n)-d}< 2^{n-K(n)-d}$. Thus the weight of our KC sets will
be no more than
$$
\sum_{n\in\omega} 2^{-(n-d)}2^{n-K(n)-d}=\sum_{n\in\omega} 2^{-K(n)}\leq1.
$$
Therefore,~$M_1$ and~$M_2$ will indeed be prefix-free machines.

In our construction, in addition to building~$M_1$ and~$M_2$, we will also
build finitely many \ce sets~$A_\ei$.  One of these sets will be our desired
set~$A$, which will be noncomputable and tt-reducible to~$R_1$ and~$R_2$.
However, we do not know which of the sets will be the true set~$A$.

Let $O^j_n[s]=\{\sigma \in 2^n \mid \KUj(\sigma)[s]<n\}$ for $j=1, 2$.  Note
that these are the strings of length~$n$ that have been shown to be outside
of~$R_j$ by stage~$s$.

\subsection{Construction}


{\em Stage $0$}.  Let $l(e,i)=0$ for all $e,i\leq 2^{c+d}$.

\medskip
{\em Stage $s+1$, Part~$1$.} For each pair~$\ei$ with $e,i\leq 2^{c+d}$, in
decreasing order (starting from the largest~$\ei$), check whether there is an
$n\leq s$ such that
\begin{enumerate}[label=(\roman*)]
\item $n$ is unused and $n>b+d+2\log n$,
\item $n\neq m_{{\eip},x'}$ for any $\eip\geq \ei$ and any~$x'$,
\item $e2^{n-g(n)-d}\leq|O^1_n[s]|$,
\item $i2^{n-g(n)-d}\leq |O^2_n[s]|$, and
\item $g(n)= K(n)[s]$.
\end{enumerate}

\noindent In other words, we check whether $g(n)= K(n)[s]$, and if so
we try to find the largest pair~$\ei$ satisfying the above criteria.

If so, then take the least such~$n$ and apply the following steps:

\begin{itemize}
\item Let $S^1_{\ei,l(e,i)}$ be the least~$k$ elements in $2^n-O^1_n[s]$,
where
$$
k=\min\{2^n-|O^1_n[s]|, 2^{n-g(n)-d}\},
$$
and similarly for $S^2_{\ei,l(e,i)}$.
\item Let $m_{\ei,l(e,i)}=n$.
\item Increment $l(e,i)$ by~1.
\end{itemize}

{\em Stage $s+1$, Part~$2$.} If $m_{\ei, x}=n$ for some~$x$, we call~$n$ a
{\em candidate} for~$\ei$. If~$n$ is an unused candidate for any~$\ei$ and
$g(n)>K(n)[s]$ (i.e.,~$K(n)$ decreased from~(v) above when~$n$ was made a
candidate), we declare~$n$ to be {\em used} and let $E^1_n=S^1_{\ei,x}$
where~$\ei$ is the greatest such that~$n$ is a candidate for~$\ei$ and~$x$ is
such that $m_{\ei, x}=n$. Similarly, we define $E^2_n=S^2_{\ei,x}$.

When~$n$ becomes used, for each~$\ei$ such that~$n$ is a candidate for~$\ei$,
we enumerate $\langle n, t\rangle$ into~$A_\ei$, where~$t$ is the stage at
which~$n$ became a candidate for~$\ei$.

End of construction.

\medskip
Let~$M_j$ be the machine that compresses each string in~$E_n^j$ to length
$n-d$, for $j=1,2$.
As explained previously, these machines are guaranteed to be prefix-free by
the KC Theorem.

To see how we can know~$d$ in advance, first note that by the Recursion
Theorem, we can know indices for the KC sets that we are building in order to
define~$M_1$ and~$M_2$.  Since we can effectively go from an index for a KC
set to an index for a machine, by the KC Theorem, we can find indices for the
machines~$M_1$ and~$M_2$.  Since~$U_1$ and~$U_2$ are universal prefix-free
machines, given indices for~$M_1$ and~$M_2$, we can effectively find coding
constants~$d_1$ and~$d_2$ such that $K_{U_j}(\sigma)\leq K_{M_j}(\sigma)
+d_j$ for each $j=1,2$ and $\sigma\in 2^{<\omega}$.  Thus, we can know~$d_1$
and~$d_2$ in advance, so let $d=d_1+d_2+1$.

Note that for almost all~$n$, $n>b+d+2\log n$.  For such~$n$,
$$
n-g(n)-d > b+d+2\log n -g(n) -d \geq 0.
$$
Thus, $2^{n-g(n)-d}\geq 1$. Let~$I$ be the infinite set of~$n$ such that
$n>b+d+2\log n$ and $g(n)=K(n)$.  For each $n\in I$, there is some~$\ei$ such
that conditions (i)-(v) will hold.  Now each $n \in I$ can be a candidate for
each~$\ei$ at most once, and~$n$ will eventually become a candidate for
some~$\ei$ since if~$n$ is not already a candidate for some~$\ei$,~$n$ will
eventually become a candidate for $e=i=0$.  Thus, there is some
largest~$\ei$, which we will call $\eizero$, such that there are infinitely
many elements of~$I$ that become candidates for~$\ei$.  Note that neither
coordinate of~$\eizero$ can be equal to $2^{c+d}$, as we know that for all
$n\in I$, since $g(n)=K(n)$, the set of compressible strings of length~$n$ is
strictly less than $2^{n-g(n)+c}$.

Let $A=A_\eizero$.  We will show that~$A$ is tt-reducible to both~$R_1$
and~$R_2$ and is not computable.  (We will not use the other sets $A_\ei$.)


\subsection{Verification}

\begin{lem}\label{ttlemma}
$A\lett \RKUj$, for $j=1, 2$.
\end{lem}

The proof depends on the following two sublemmas.

\begin{slem}
\label{Esublemma} For all~$n$ such that there exists~$x$ with $\meizerox=n$,
we have that $g(n)>K(n)$ implies $E^j_n=S^j_{\ei,y}$ for some $\ei\geq
\eizero$ and some $y\in \omega$, and that $g(n)=K(n)$ implies
$E^j_n=\emptyset$ for $j=1, 2$.
\end{slem}

\begin{proof}
In the construction, if $g(n)>K(n)$,~$n$ will become used, and~$E^j_n$ will
be defined.  Since~$n$ is a candidate for $\eizero$, it must become a
candidate when it is still unused, so by the time it becomes used, the
greatest~$\ei$ for which~$n$ is a candidate is at least~$\eizero$.
Thus,~$E^j_n$ will be defined as $S^j_{\ei,y}$ for some $\ei\geq \eizero$. If
$g(n)=K(n)$, then~$n$ will never become used, and~$E^j_n$ will never be
nonempty.
\end{proof}

\begin{slem}\label{containment}
For almost all~$x$, for each $j=1,2$,
$$
S^j_{\eizero,x}\subseteq \overline{\RKUj}\iff
    \langle \meizerox, s\rangle \in A,
$$
where~$s$ is the stage at which $\meizerox$ was defined.
\end{slem}

\begin{proof}
$(\Leftarrow)$ Since $\langle n,s\rangle=\langle \meizerox, s\rangle \in A$,
we have $g(n)>K(n)$.  By  Sublemma~\ref{Esublemma}, $E^j_n=S^j_{\ei,y}$ for
some $\ei\geq \eizero$.   Since all strings in~$E^j_n$ are compressed
by~$M_j$ to length $n-d$, they are compressed by~$\KUj$ for $j=1, 2$ to
length less than $n-d+d=n$.    Thus, $S^j_{\ei,y}\subseteq \overline{\RKUj}$.
If $\ei=\eizero$, we are done, since in that case, $x=y$.  Otherwise, $\ei>
\eizero$, so $S^j_{\ei,y}$ must have been defined after $S^j_{\eizero, x}$
was defined because otherwise condition (ii) of the construction would not
allow $\meizerox$ to be defined as~$n$.
Thus, $S^j_{\eizero, x}\subseteq \overline{\RKUj}$ as well, because anything
in $S^j_{\eizero, x}-S^j_{\ei,y}$ must have already been seen to be nonrandom
with respect to~$R_j$ by the time $S^j_{\ei,y}$ was defined, since such
strings are in $O^j_n=\ov{R_j}\cap 2^n$.

$(\Rightarrow)$ Let~$x_0$ be such that for all $x\geq x_0$, if $\meizerox=n$
for some $n\in I$, then~$n$ never becomes a candidate for any $\ei>\eizero$.
Let $x\geq x_0$ and $n=\meizerox$. Let $S^j_{\eizero,x}$ become defined at
stage~$s$. Suppose $S^j_{\eizero,x}\subseteq\overline{\RKUj}$ for either
$j=1$ or~$2$. Without loss of generality, assume $j=1$.
Then at some stage $t>s$, all strings in $S^1_{\eizero,x}$ are in $O^1_n[t]$.
This means $|O^1_n[t]|\geq |O^1_n[s]|+k$, where~$k$ is the number of strings
in $S^1_{\eizero,x}$, all of which still appeared random at stage~$s$. Recall
that the number of elements in $S^1_{\eizero,x}$ was chosen to be the minimum
of $2^n-|O^1_n[s]|$ and $2^{n-g(n)-d}$. In the former case, all strings of
length~$n$ are non-random, which is impossible. So $S^1_{\eizero,x}$ has size
$2^{n-g(n)-d}$, and by condition (iii), $|O^1_n[s]|\geq {e_0}2^{n-g(n)-d}$,
so $|O^1_n[t]|\geq (e_0+1)2^{n-g(n)-d}$. We also have that $|O^2_n[t]|\geq
{i_0}2^{n-g(n)-d}$, so condition (iii) and (iv) will hold for some~$\ei$ with
$\ei>\eizero$. If~$n$ has not yet become used when this happens, all of
conditions (i)-(v) will hold for this~$\ei$ and~$n$, so~$n$ will eventually
become a candidate for~$\ei$. Thus, $n\notin I$. Therefore, $g(n)>K(n)$, and
so at some stage in the construction,~$n$ will become used while it is a
candidate for~$\eizero$, and at this point $\ns=\langle \meizerox, s\rangle$
will be enumerated into~$A$.
\end{proof}

\begin{proof}[Proof of Lemma~\ref{ttlemma}]
To determine if~$\ns$ is in~$A$, run the construction to see if $\meizerox$
is defined as~$n$ for some~$x$ at stage~$s$.  If not, then $\ns\notin A$.  If
so, then $\ns\in A\iff S^j_{\eizero,x}\cap \RKUj=\es$.  This works for both
$j=1$ and~$2$ by Sublemma~\ref{containment}.
\end{proof}

\begin{lem} $A$ is not computable.
\end{lem}

\begin{proof}
Suppose~$A$ is computable.  Let
$$
B=\{\ns \mid \text{$n$ becomes a candidate for $\eizero$ at stage~$s$}\}.
$$
Obviously~$B$ is computable, and so $B-A$ is computable. $B-A$ is the set of
all candidate pairs $\ns$ such that $n\in I$. Let $C=\{n \mid (\exists s)[\ns
\in B- A]\}$. Then~$C$ is an infinite \ce set such that $C\subseteq I$
(here~$C$ is infinite by the choice of $\eizero$). However,
Solovay~\cite{Solovay} showed that~$I$ contains no infinite \ce set, so this
is a contradiction. Thus~$A$ is not computable, proving the theorem.
\end{proof}

This proof can easily be modified to accommodate any finite set of universal
machines by replacing the pairs~$\ei$ with $m$-tuples.  This gives the
following:

\begin{thm}\label{thmfinite}
For any finite set of prefix-free universal machines $\{U_j\}_{j=1,\ldots,m}$
there is a noncomputable \ce set~$A$ such that $A\lett R_{K_{U_j}}$ for each
$j=1,\ldots, m$.
\end{thm}

While sets of random strings cannot form a tt-minimal pair, there are Turing
complete sets that do form a tt-minimal pair, as we show in the following
section.

\section{A tt-minimal pair of Turing complete sets}\label{minpair}

In Theorem~\ref{theorem1}, we showed that there is no pair of sets of random
strings $R_{K_{U_1}}$ and $R_{K_{U_2}}$ that form a minimal pair in the
truth-table degrees, or even in the \ce truth-table degrees. We know that
$R_{K_U}$ is always Turing complete. If no two Turing complete sets ever form
a minimal pair in the tt-degrees, Theorem~\ref{theorem1} would be a trivial
corollary. However, this is not the case, as we show in this section. By a
different method, Degtev~\cite{Degtev} proved that there are Turing complete
\ce sets that form a minimal pair in the \emph{c.e.\ }truth-table degrees. We
produce a minimal pair in the full structure of the tt-degrees.

\begin{thm}\label{minpairthm}
There exist Turing complete c.e.\ sets~$A_1$ and~$A_2$ whose tt-degrees form a
minimal pair.
\end{thm}

\begin{proof}
Let $\{\Psi_i\}_{i\in\omega}$ be a computable listing of all partial
truth-table reductions.

Let~$D$ be a Turing complete \ce set with a computable enumeration
$\{D_s\}_{s\in\omega}$ such that if~$n$ enters~$D$ at stage~$s$, then~$m$
enters~$D$ at stage~$s$ for all $m\in [n,s)$ not yet in~$D$.  Such a set~$D$
can be constructed using a standard movable marker construction.
 We will build~$A_1$ and~$A_2$ as well as Turing functionals~$\Gamma_1$
and~$\Gamma_2$, such that $\Gamma_i^{A_i}=D$ for $i=1,$ 2, satisfying the
following requirements for all~$e$:
\begin{align*}
\mathcal R_e &: \Psi_e^{A_1}=\Psi_e^{A_2}=f \text{ total } \implies
 \text{ $f$ computable}.
\end{align*}
By Posner's trick, the requirements~$\mathcal R_e$ (and the fact
that $A_1 \neq A_2$) suffice to show that the tt-degrees of~$A_1$ and~$A_2$
form a minimal pair, because if $\Psi_i^{A_1}=\Psi_{i'}^{A_2}$, we could
build a single tt-reduction~$\Psi_e$ such that $\Psi_e^{A_1}=\Psi_i^{A_1}$
and $\Psi_e^{A_2}=\Psi_{i'}^{A_2}$.

We will build $\Gamma_i^{A_i}$ in stages with uses $\gamma_i(x,s)$, for
$i=1,$ 2.  In particular, we will treat the use $\gamma_i(x,s)$ as a movable
marker.  The marker $\gamma_i(x,s)$ sits on an element not yet in~$A_i$.  We
may change the value of $\Gamma_i^{A_i}$ by enumerating $\gamma_i(x,s)$
into~$A_i$.  The movement of the markers is subject to the following rules:

\begin{enumerate}
\item If $n<n'$, then $\gamma_i(n,s)<\gamma_i(n',s)$.
\item $\gamma_i(n,s+1)\ne \gamma_i(n,s)$ implies $\gamma(n,s+1)>s+1$,
    where by convention~$s$ exceeds all numbers used in computations at
    stage~$s$. We refer to this action as \emph{kicking}. Moreover, when
    $\gamma_i(s,s)$ is first appointed at the end of stage~$s$, it is
    chosen to be the least element not yet in~$A_i$ that is greater
    than~$s$ and all other markers $\gamma_i(n,s)$.
\item If~$n$ enters~$D$ at~$s$ then we will  enumerate $\gamma_i(n,s)$
    into $A_i[s+1]$.
Once $n\in D_s$ we will no longer define $\gamma_i(n,s)$.  The marker
will be removed.
\item If $\gamma_i(n,s)$ enters $A_i[s]$, so do all currently defined
    $\gamma_i(k,s)$ for all $k\in [n,s)$.
\item Coding of~$D$ is not the only reason $\gamma_i(n,s)$ can move. The
    marker $\gamma_i(n,s)$ may be moved by requirements~$\mathcal R_e$ in
    their attempts to seek satisfaction, but~$\mathcal R_e$ can only move
    $\gamma_i(n,s)$ if $e\le n$. As we will show, a single~$\mathcal R_e$
    can only move a specific $\gamma_i(n,s)$ a finite number of times. If
    $\gamma_i(n,s)$  enters~$A_i$ and $n\not \in D_s$, it will be
    redefined, and as usual, we will kick $\gamma_i(n,s+1)$ to a fresh
    element past $s+1$.
\item If $\gamma_i(n,s)$ enters~$A_i$ then one of $\gamma_j(n,s)$ or
    $\gamma_j(m,s)$ must simultaneously enter~$A_j$ for $j\ne i$, where
    $m>n$ is the smallest $\gamma_j(m,s)$ with $m\not\in D_s$.  That is,
    $\gamma_j(m,s)$ is the least marker still defined for $m>n$.
\item If $\gamma_i(n,s)$ moves or is enumerated, then~$\mathcal R_e$ is
    initialized for $e>n$, meaning that all current values for $f=f_e$
    are discarded, and the strategies for~$\mathcal R_e$ are restarted.
\end{enumerate}

\noindent To achieve~$\mathcal R_e$, we will force disagreements at stage~$s$ between
$\Psi_e^{A_1}$ and $\Psi_e^{A_2}$ whenever possible by enumerating
$\gamma_1(n,s)$ into~$A_1$ and $\gamma_2(m,s)$ into~$A_2$, where~$n$ and~$m$
are at most one defined marker apart, as specified in Rule~6.  According to
the rules that govern marker movement, we also enumerate all larger markers.

Let $\ell(e,s)$ be the length of agreement function given by
$$
\ell(e,s)=\max\{n \mid \Psi_e^{A_1}\uh n[s]=\Psi_e^{A_2}\uh n[s]\}.
$$
If the limit of $\ell(e,s)$ is infinity, then we will try to fix the values
of $f=\Psi_e^{A_i}$ so that~$f$ is computable.  Given values for $f\uh n$, we
will attempt to force a disagreement between $\Psi_e^{A_1}(n)$ and
$\Psi_e^{A_2}(n)$ while following the rules of marker movement.  Any
disagreement we force could only be injured finitely often, and we will
eventually either preserve a disagreement or reach a believable computation
for~$f(n)$.

We perform the construction on a tree of strategies. Each height~$e$ will
correspond to the strategy~$\mathcal R_e$. Nodes of length~$e$ will be
extended by the three possible outcomes for the strategy~$\mathcal
R_e$:~$\infty$,~d (for disagreement), and~w (for waiting), ordered by
$\infty< $ d $ <$ w. The~$\infty$ outcome will correspond to the situation
where $\Psi_e^{A_1} =\Psi_e^{A_2}$. The d outcome will correspond to the
situation where we are preserving a disagreement between $\Psi_e^{A_1}$ and
$\Psi_e^{A_2}$. Otherwise, the outcome will be w; this includes the case in
which $\Psi_e$ is not a total truth-table reduction and our strategy is
eventually stuck waiting for convergence.

We first discuss the basic module for~$R_0$.  We will then modify this to the
$\alpha$-module by giving a formal construction in Section
\ref{minpairconst}.

\noindent
\subsection{\texorpdfstring{Basic module for $\mathcal R_0$}%
{Basic module for R0}}\label{basic}

The module works in order of~$k$ to give a definition of~$f(k)$.

For $k=0$:

\begin{itemize}
\item Wait till the first stage~$s$ when $\ell(0,s)>0$. Immediately
    enumerate $\gamma_i(q,s)$ into~$A_i$ for all $0<q<s$. This causes
    $\gamma_i(q,s+1)$ to be moved past the use of $\Psi_0(0)$, so that
    enumeration into~$A_i$ will not affect the computations.  We will say
    that the pair $\<0,k\> =\langle 0, 0\rangle$ has been {\em prepared}.
\item Wait until the next stage~$t$ where $\ell(0,t)>0$. At this stage
    there are two possibilities.
\begin{enumerate}
\item Putting $\gamma_1(0,t)$ into~$A_1$, or $\gamma_2(0,t)$
    into~$A_2$ or both, will cause a disagreement at argument~0.
\item Otherwise. Then define $f(0)=\Psi_0^{A_1}(0)[t].$

\end{enumerate}
\item Suppose we invoke~1. If we put both of the markers into their
    targets~$A_1$ and~$A_2$, then the strategy is successful by kicking
    because no markers will ever be defined below the use of $\Psi_0(0)$
    and thus our disagreement can never be injured. If we only changed
    one side, say~$A_1$,  then this will cause a disagreement that holds
    forever, unless at a later stage~$t'$,~$0$ enters~$D$. At such a
    stage~$t'$ we would enumerate $\gamma_i(0,t')$ into~$A_i$, and noting
    that $\gamma_2(0,t')=\gamma_2(0,t)$, this could potentially make the
    computation equal again. We would wait until the next stage $t''>t'$
    where $\ell(0,t'')>0$, and define $f(0)=\Psi_0^{A_1}(0)[t'']$, safe
    in the knowledge that this is now an immutable computation.
\end{itemize}

\noindent Given $f(k)$, we act for $k+1$:
\begin{itemize}
\item After defining $f(k)$, we wait for the stage~$s$ where
    $\ell(0,s)>k+1$. We then enumerate all $\gamma_i(q,s)$ into~$A_i$ for
    $q>k+1$. As before, this causes $\gamma_i(q,s+1)$ to move past
    $\psi_0(k+1)$, the use of $\Psi_0(k+1)$, and we call $\< 0, k+1\rangle$
    {\em prepared}.
\item Wait until the next stage~$t$ where $\ell(0,t)> k+1$. We examine
    the tt-reductions $\Psi_0(k+1)$ and allowable enumerations of
    $\gamma_i(n,t)$ into~$A_i$ for $n\le k+1$ below the use,
    $\psi_0(k+1)$, to see if we can cause a disagreement for argument
    $k+1$. Again by kicking, everything else is too big.  If we can cause
    a disagreement, we will do so with the least possible elements.  To
    be more specific, given $m\in\omega$, let~$m^-$ be the greatest
    $m'<m$ such that $\gamma_0(m',t)$ is defined and let~$m^+$ be the
    least $m''>m$ such that $\gamma_0(m'',t)$ is defined. By the rules of
    movement,~$m^-$ is the greatest number less than~$m$ such that
    $m^-\notin D_t$ and~$m^+$ is the least number greater than~$m$ such
    that $m^+\notin D_t$. Let~$m$ be the least element such that we can
    cause a disagreement by enumerating $\gamma_1(m,t)$ into~$A_1$ and
    either $\gamma_2(m^-,t)$, $\gamma_2(m,t)$, or $\gamma_2(m^+,t)$
    into~$A_2$, as well as all greater markers, according to the rules of
    movement.  We choose the least pairing that causes a disagreement and
    enumerate the appropriate elements.  Again, there are two
    possibilities:
\begin{enumerate}
\item We make such an enumeration to cause a disagreement.
When we implement the tree of strategies, nodes guessing that there
is a disagreement at~$\mathcal R_0$ will preserve the disagreement.
\item No such~$m$ exists to cause a disagreement. Then define
    $f(k+1)=\Psi_0^{A_1}(k+1)[t].$
\end{enumerate}
\end{itemize}

\noindent In case~1, the disagreement at $k+1$ may be injured. $\mathcal R_0$ will not
act when it sees a disagreement, so injury can only occur by elements
entering~$D$. If such elements do enter~$D$, causing an agreement between
$\Psi_0^{A_1}(k+1)$ and $\Psi_0^{A_2}(k+1)$, we will wait until we see
$\ell(0,s)>k+1$ and will try again to cause a disagreement. When we cannot,
we will define $f(k+1)=\Psi_0^{A_1}(k+1)[s]$. In fact, we will not be able to
find a new disagreement because we previously chose the minimal possible
disagreement.


\subsection{Tree of strategies}\label{treeofstrategies}

As mentioned previously, each node on our tree of strategies will be extended
by three possible outcomes:~$\infty<$ d $<$ w, where the ordering is left to
right.  We will build an approximation to the true path through the tree,
which we call~$\delta_s$.  We say~$s$ is an {\em $\alpha$-stage} if~$\alpha$
is a prefix of~$\delta_s$.  Nodes~$\alpha$ of length~$e$ can act for
$\mathcal R_e$ only at $\alpha$-stages.  During such action, any attempt at
defining the function given by $\Psi_e^{A_1}=\Psi_e^{A_2}$ will be
called~$f_\alpha$.  Whenever~$\delta_s$ moves to the left
of~$\alpha$,~$\alpha$ will be initialized, undefining all values
of~$f_\alpha$.  For~$\alpha$ on the true path, which is  $\liminf_s
\delta_s$, this will only happen finitely often.

We build an approximation~$\delta_s$ to the true path recursively as follows:
Given $\alpha=\delta_s\res (e+1)$, we define $\delta_s(e+1)$.
If $\ell(e,s)$ is greater than it has been at any previous $\alpha$-stage,
then $\delta_s(e+1)=\infty$.
If we have acted at some stage $t\leq s$ to cause a disagreement between
$\Psi_e^{A_1}(k)$ and $\Psi_e^{A_2}(k)$ and this disagreement has been
preserved, then $\delta_s(e+1)=$ d. Otherwise $\delta_s(e+1)=$ w, the waiting
outcome. We define~$\delta_s$ in this way until we have defined
$\delta_s(s-1)$, so that~$\delta_s$ has length~$s$.

We will not allow any disagreements to be injured by nodes extending or to
the right of the ``d'' outcome.  To achieve this, we will only allow each
node $\alpha=\delta_s\res e$ to enumerate elements $\gamma_i(n,s)$ for~$n$
greater than or equal to the last stage~$s_\alpha$ such that the
approximation to the true path was to the left of~$\alpha$.  If, for $e'<e$,
$\delta_s(e')=$ d, preserving a disagreement at~$k$, then the last stage~$t$
such that $\delta_t(e')=\infty$ must have been a stage where the truth-table
for $\Psi_{e'}(k)$ had already been defined, since we will not act to cause a
disagreement at~$k$ until we first see $\ell(e',s)>k$.  By convention, any
stage at which the truth-table for $\Psi_{e'}(k)$ has been defined must be
greater than the use $\psi_{e'}(k)$, so $s_\alpha>t>\psi_{e'}(k)$. Similarly,
if $\delta_s(e')=$ w, then the last stage such that the true path went
through d or~$\infty$ was larger than the use $\psi_{e'}(k)$, for~$k$ the
last spot where we caused a disagreement.


\subsection{Construction}\label{minpairconst}

{\em Stage $0$.}  Let $A_1[0]=\emptyset$ and $A_2[0]=\{0\}$.  Let
$\delta_0=\lambda$, the empty string.  Define $\gamma_i(0,0)=1$ for $i=1,$ 2.
Note that we have guaranteed that $A_1\neq A_2$.

{\em Stage $s+1$.}

Suppose~$n$ enters~$D$ at stage $s+1$.  Enumerate $\gamma_i(n,s)$ into
$A_i[s+1]$ for $i=1,$ 2.  We remove the marker $\gamma_i(n,s)$, so we will
not define $\gamma_i(n,s+1)$.  Initialize all~$\mathcal R_e$ for $e>n$,
undefining any values of~$f_\alpha$ for $|\alpha|=e$.

In increasing order of~$e$, for every $e\leq s$, do the following:

Let $\alpha=\delta_s\res e$, where outcomes are as described in
Section~\ref{treeofstrategies}.  Let $s_\alpha$ be the last stage~$t$ such
that~$\delta_t$ was to the left of~$\alpha$, or~0 if the approximation to the
true path has never been to the left of~$\alpha$.  Let~$k$ be the greatest
such that $f_{\alpha}(k)$ is defined, or ${-1}$ if there is no such~$k$.

{\em Step 1: Preparing $\langle e,k+1\rangle$.} If $\langle e,
k+1\rangle$ has never before been prepared, and $\ell(e,s)>k+1$ for
the first time since defining $f_\alpha(k)$, enumerate all
$\gamma_i(q,s)$ into~$A_i$ for all $q$ satisfying $q\geq\max\{e, k+2,
s_\alpha\}$.  We say the we have now {\em prepared} the pair $\langle
e, k+1\rangle$.  Move each $\gamma_i(q,s)$ (that is still defined), in
increasing order of~$q$, to the next fresh spot greater than $s+1$
according to the rules of motion.  This will prevent $\gamma_i(q,s)$
from influencing $\Psi_e^{A_i}(k+1)$ since it has been kicked past
$\psi_e(k+1)$. Initialize all~$\mathcal R_{e'}$, for $e'>q$, as in
Rule~7.  Note that we will call these newly kicked markers
$\gamma_i(q,s)$ until the end of the stage, where all markers will be
renamed to $\gamma_i(q,s+1)$.  If we prepared some pair in this step,
begin the steps for $e+1$.  Otherwise, go to Step~2.

{\em Step 2: Searching for a disagreement.}  If $\ell(e,s)>k+1$, we will
attempt to cause a disagreement at $k+1$.  Let~$m^-$ and~$m^+$ be as defined
in the basic module in Section \ref{basic}.  Let $m\geq \max\{e, s_\alpha\}$
be the least element such that we can cause a disagreement between
$\Psi_e^{A_1}(k+1)$ and $\Psi_e^{A_2}(k+1)$ by enumerating $\gamma_1(m, s)$
into~$A_1$ and either $\gamma_2(m^-, s),\ \gamma_2(m,s),$ or $\gamma_2(m^+,
s)$ into~$A_2$, as well as all larger markers.  (Of course, we do not consider
enumerating $\gamma_2(m^-, s)$ unless $m^-\geq \max\{e, s_\alpha\}$.)
We choose the least pairing
that causes a disagreement and enumerate the pair and all larger markers into
the corresponding~$A_i$.  We move all enumerated markers to the next fresh
spots greater than $s+1$.  If we were unable to cause a disagreement, we
define $f_\alpha(k+1)=\Psi_e^{A_1}(k+1)[s]$.

Add a new marker $\gamma_i(s+1, s+1)$ to the first fresh spot greater than
$s+1$.  Note that this $\gamma_i(s+1, s+1)$ will be greater than the current
(and former) locations of all other markers.
For any node~$\beta$ to the right of $\delta_s$, initialize~$\beta$ by
undefining all values of~$f_\beta$.

\subsection{Verification}\label{minpairverif}

Let the true path of the construction be $\liminf_s \delta_s$.

\begin{lem}\label{finitemovement}
Each marker moves finitely often.  That is, for $i=1,\ 2$ and $k\in\omega$,
there are finitely many stages~$s$ such that $\gamma_i(k,s)\neq
\gamma_i(k,s+1)$.
\end{lem}

\begin{proof}
Induct on~$k$.  Suppose the lemma is true for all $n\leq k$ and $i=1, \ 2$.
We will show it holds for $k+1$ as well.  If $k+1\in D$, then when $k+1$ is
enumerated into~$D$, $\gamma_i(k+1,s)$ is enumerated into~$A_i$ and the
marker is not redefined.  Thus $\gamma_i(k+1,s)$ moves finitely often.  So we
will assume that $k+1\notin D$.

According to Rule~5,~$\mathcal R_e$ can only move $\gamma_i(k+1,s)$ if $e\leq
k+1$.  Thus it is enough to show that none of these~$\mathcal R_e$ moves
$\gamma_i(k+1)$ infinitely often.  Suppose for a contradiction that~$\mathcal
R_e$ moves $\gamma_i(k+1)$ infinitely often and that~$e$ is the least such
that this happens for either~$i$.  There are two ways~$\mathcal R_e$ could
move $\gamma_i(k+1, s)$.  First, by preparing $\langle e, n\rangle$ for some
$n\leq k$ as in Step~1 of the Construction.  Since each pair $\langle
e,n\rangle$ can only be prepared once, this can only happen finitely often.


The other way that~$\mathcal R_e$ can move $\gamma_i(k+1,s)$ is by action of
Step~2 in the Construction, causing a disagreement.  By induction, there is a
stage~$t_1$ after which no markers $\gamma_i(n,s)$ ever move or are removed
for $n\leq k$.  By the minimality of~$e$, there is a stage~$t_2$ after which
no~$\mathcal R_{e'}$ moves $\gamma_i(k+1,s)$ for any $e'<e$. By the previous
paragraph, there is a stage~$t_3$ such that Step~1 of~$\mathcal R_e$ has
stopped moving $\gamma_i(k+1,s)$ by stage~$t_3$.  Let~$\alpha$ be the node of
length~$e$ on the true path.  Let~$t_4$ be a stage by which~$\delta_s$ never
goes to the left of~$\alpha$ after stage~$t_4$.  Finally, let $t>t_1, \ t_2,
\ t_3$, and~$t_4$.

Note that since~$\mathcal R_e$ acts infinitely often by moving
$\gamma_i(k+1,s)$, it must do so at infinitely many $\alpha$-stages,
for~$\alpha$ the length~$e$ node on the true path.  This is because
each~$\alpha'$ to the right of~$\alpha$ can only move elements greater than
the last stage at which they were initialized, and they will be initialized
infinitely often since they are not on the true path.

Now suppose that at some $\alpha$-stage $s_0\geq t$,~$\mathcal R_e$ acts by
enumerating $\gamma_i(k+1,s_0)$ to cause a disagreement between
$\Psi_e^{A_1}(n)$ and $\Psi_e^{A_2}(n)$ for some~$n$.  By assumption,
$\mathcal R_e$ will eventually act again at an $\alpha$-stage by enumerating
$\gamma_i(k+1,s)$ to cause a disagreement between $\Psi_e^{A_1}(n')$ and
$\Psi_e^{A_2}(n')$ for some $n'>n$.  This means that at some stage $s>s_0$,
$\ell(e,s)>n$, so the disagreement achieved at stage~$s_0$ will be injured.

We must examine how such an injury could happen.  Since the disagreement was
caused by enumerating $\gamma_i(k+1,s_0)$, we also must have enumerated
\mbox{$\gamma_j((k+1)^+,s_0)$,} for $j\neq i$, by Rule~6.  If we also
enumerated $\gamma_j(k+1,s_0)$ itself, then injury would be impossible since
the only markers still below the use of $\Psi_e(k+1)$ have stopped moving by
stage~$t_1$.  Thus, we must not have enumerated $\gamma_j(k+1,s_0)$ and
instead enumerated the marker succeeding it.  The only way the computation
can be injured is for $\gamma_j(k+1,s)$ to be enumerated. This cannot be
enumerated by any higher priority~$e'$ or any~$\alpha'$ to the left
of~$\alpha$, by the choice of~$t_2$ and~$t_4$.  It also cannot be enumerated
by any node to the right of~$\alpha$ because such a node will not be able to
move elements smaller than the last $\alpha$-stage, which must have been
bigger than the use of $\Psi_e(k+1)$.  Any node extending $\alpha\concat$d or
$\alpha\concat$w must also preserve the disagreement because~$\alpha$ must
have been extended by~$\infty$ at some stage after the truth-table for
$\Psi_e(k+1)$ was defined, and nodes cannot move elements smaller than the
last stage at which they were initialized.  In addition, we may ignore any
node extending $\alpha\concat\infty$ because we would not go to that outcome
unless the disagreement in question had already been injured.  Thus, there is
no way for the disagreement to be injured, and~$\mathcal R_e$ will never act
again at an $\alpha$-stage, contradicting our assumption that it would act
infinitely often.

Since~$\mathcal R_e$ cannot move $\gamma_i(k+1,s)$ infinitely often for
any~$e$, we can see that $\gamma_i(k+1,s)$ can only move finitely often.
Thus, by induction, each marker only moves finitely often.
\end{proof}

\begin{lem}
$D \leT A_1, A_2$.
\end{lem}

\begin{proof}
For $i=1$ or~$2$, to compute~$D(n)$, run the construction until the first
stage $s>n$ such that either $n\in D_s$ or $A_i[s] \res (\gamma_i(n,s)+1)
=A_{i} \res (\gamma_i(n,s)+1)$.  Such a stage exists because
$\gamma_i(n,s)$ can only move finitely often.  Now $n \in D$ if and only if
$n \in D_s$.  This is because, when~$n$ enters~$D$, $\gamma_i(n,s)$ is
enumerated into~$A_i$ before it is moved.
\end{proof}

\begin{lem}\label{Satisfied}
Requirement~$\mathcal R_e$ is satisfied for each $e\in \omega$.  That is, if
there is a total function~$f$ such that \mbox{$\Psi_e^{A_1}=\Psi_e^{A_2}=f$,}
then~$f$ is computable.
\end{lem}

\begin{proof}
Suppose $\Psi_e^{A_1}=\Psi_e^{A_2}=f$ total.  Then $\lim_s \ell(e,s)=\infty$.
Let~$\alpha$ be the node of length~$e$ on the true path.   We will show that
along the true path, for almost all~$k$, $f_\alpha(k)=\Psi_e^{A_i}(k)$.

Let~$s_\alpha$ be the greatest stage such that $\delta_{s_\alpha}$ is to the
left of~$\alpha$.  Then by the construction, after
stage~$s_\alpha$,~$\mathcal R_e$ will not be allowed to enumerate any
$\gamma_i(n,s)$ for $n< s_\alpha$.  Let $s'_\alpha>s_\alpha$ be a stage such
that $\gamma_i(n,s)$ has stopped moving by stage~$s'_\alpha$ for all $n<
\max\{e, s_\alpha\}$.  After this stage, the~$f_\alpha$ that we are building
will be the final~$f_\alpha$. Let~$k_0$ be the greatest~$k$ such that
$f_\alpha(k)$ was defined before stage~$s'_\alpha$ for the final~$f_\alpha$.
We will show that for $k>k_0$, $f_\alpha(k)=\Psi_e^{A_i}(k)$.

Suppose $f_\alpha(k)\neq \Psi_e^{A_i}(k)$ for some $k> k_0$.  Choose the
least such~$k$.  Suppose $f_\alpha(k)$ is defined at stage~$s$.  After
stage~$s$, some element enters~$A_1$ or~$A_2$ below the use $\psi_e(k)$.  At
some prior stage~$s'$, $\langle e,k\rangle$ was prepared as in Step~1 of the
Construction, kicking all $\gamma_i(n,s')$ for $n>k$ past~$s'$, which is
greater than $\psi_e(k)$.  Thus no $\gamma_i(n,t)$ for $n>k$ could enter
either~$A_i$ below $\psi_e(k)$ for $t\geq s$.  Therefore, any injury to the
current values of the $\Psi_e^{A_i}(k)[s]$ must be caused by some
$\gamma_i(n,t)$ entering~$A_i$ at stage $t\geq s$ for either~$i$, where~$n$
satisfies $\max\{e,s_\alpha\}\leq n\leq  k$ and $\gamma_i(n,t)<\psi_e(k,s)$.
Such $\gamma_i(n,t)$ are the only markers that both would be allowed to
enter~$A_i$ and would be able to cause injury.

\begin{clm}\label{cl:dis}
If we can cause a disagreement between $\Psi_e^{A_1}(k)$ and
$\Psi_e^{A_2}(k)$ at stage $t\geq s$, then we could have caused a
disagreement at stage~$s$ instead of defining~$f_\alpha(k)$.
\end{clm}

\begin{proof}
Suppose enumerating $\gamma_i(m,t)$ and $\gamma_j(m',t)$ as well as all
greater markers, causes a disagreement between $\Psi_e^{A_1}(k)$ and
$\Psi_e^{A_2}(k)$, where $i\neq j$ and $m'\leq m$.  Note that at least one of
$\gamma_i(m,t)$ and $\gamma_j(m',t)$ must be below the use of $\Psi_e(k)$,
hence less than~$s$.  Any marker that is at a position less than~$s$
at stage~$t$ will have been at the same position at stage~$s$, because no
markers are moved or added to numbers below~$s$ at or after stage~$s$.

\smallskip
{\em Case 1:} $m'=m$.  If both markers $\gamma_1(m,t)$ and $\gamma_2(m,t)$
are in the same spots as $\gamma_1(m,s)$ and $\gamma_2(m,s)$, then the same
enumeration could have been made to cause a disagreement instead of defining
$f_\alpha(k)$.  Suppose $\gamma_2(m,t)\neq \gamma_2(m,s)$.  Then since one
marker must be below the use, $\gamma_1(m,s)=\gamma_1(m,t)$.  Between
defining~$f_\alpha(k)$ and stage~$t$, $\gamma_2(m,s)$ moved, but since
$\gamma_1(m,s)$ didn't move, $\gamma_2(m^-,s)$ couldn't have moved, by
Rule~6.  Thus, enumerating $\gamma_1(m,s)$ and $\gamma_2(m,s)$ at stage~$s$
will give the same disagreement caused by enumerating $\gamma_i(m,t)$ for
both $i=1,2$, so we would have made this enumeration instead of defining
$f_\alpha(k)$ at stage~$s$.

\smallskip
{\em Case 2:} $m'=m^-$ at both stage~$t$ and stage~$s$.  As in Case~1, if
both markers are at the same numbers at stage~$t$ as they were at stage~$s$,
then the same enumeration could have been made instead of
defining~$f_\alpha(k)$. It is not possible that $\gamma_j(m^-,t)\neq
\gamma_j(m^-, s)$, because any movement would have forced $\gamma_i(m,s)$ to
move as well, by Rule~6, pushing it past the use.  Suppose $\gamma_i(m,s)\neq
\gamma_i(m,t)$ and $\gamma_j(m^-,s)=\gamma_j(m^-,t).$  Then the least element
that was enumerated into~$A_i$ and moved after stage~$s$ was either
$\gamma_i(m,s)$ or $\gamma_i(m^-,s)$.  Thus, at stage~$s$, we could enumerate
the appropriate one of $\gamma_i(m,s)$ or $\gamma_i(m^-,s)$ along with
$\gamma_j(m^-,s)$ to cause the same disagreement instead of defining
$f_\alpha(k)$.

\smallskip
{\em Case 3:} $m'=m^-$ at stage~$t$, but not at stage~$s$.  Then between
stage~$s$ and stage~$t$, some elements~$n$, $m'<n<m$ entered~$D$.  For the
least such~$n$, $m'=n^-$ at stage~$s$, so we could have enumerated
$\gamma_i(n,s)$ and $\gamma_j(m', s)$ to cause the same disagreement at
stage~$s$ instead of defining $f_\alpha(k)$.



Thus, any disagreement we could cause after defining $f_\alpha(k)$ could have
been caused {\em instead} of defining $f_\alpha(k)$.
\end{proof}

According to Claim~\ref{cl:dis}, in order to cause an injury to the agreement
between $f_\alpha(k)$, $\Psi_e^{A_1}(k)$ and $\Psi_e^{A_2}(k)$, the
enumeration must have caused a change in both $\Psi_e^{A_1}(k)$ and
$\Psi_e^{A_2}(k)$ to cause a new agreement between them that differs from
$f_\alpha(k)$.  Consider the greatest possible enumeration that would have
caused such a change. Suppose that the least elements of the greatest
enumeration are $\gamma_i(m,t)$ and $\gamma_j(m', t)$ for $j\neq i$ and
$m'\leq m$.  Then $\gamma_i(m,t)$ and $\gamma_j(m^+, t)$ would also be an
allowed enumeration. It could not be true that under such an enumeration,
$\Psi_e^{A_i}(k)=\Psi_e^{A_j}(k)$, as this would contradict that the pair
$\langle m,m'\rangle$ gave the greatest possible enumeration that changed
both computations to cause agreement again.  Thus, under this new
enumeration, a disagreement is caused between the two computations.  This is
impossible, since the existence of such a disagreement would have led to us
forcing the disagreement instead of defining $f_\alpha(k)$, as shown in
Claim~\ref{cl:dis}.  Thus, there can be no greatest enumeration to cause a
change in values of $\Psi_e^{A_i}(k)$, so the values will not change,
and~$f_\alpha$ was correct.  Since~$f_\alpha$ is a computable function, so is
$\Psi_e^{A_1}=\Psi_e^{A_2}$.
\end{proof}

This concludes the proof of Theorem~\ref{minpairthm}.
\end{proof}

Degtev~\cite{Degtev} and Marchenkov~\cite{Marchenkov} showed there is a \ce
tt-degree minimal among the tt-degrees; that is, there is a \ce set~$B$ such
that for all~$A$ such that $A<_{\rm tt} B$, $A$~is computable.  However, all
such tt-degrees are low$_2$, as shown by Downey and Shore~\cite{DS}. Thus,
there is no Turing complete \ce set of minimal tt-degree.

Our theorem cannot be extended to show the existence of a minimal pair of
Turing complete \ce sets within the bT-degrees (also known as wtt-degrees) by
the following

\begin{thm}[Ambos-Spies~\cite{Ambos-Spies}]
A \ce set is half of a minimal pair in the Turing degrees if and only if it
is half of a minimal pair in the bounded-Turing degrees.
\end{thm}

Thus, no \ce Turing complete set is half of a minimal pair in the bT-degrees.
In contrast, our Theorem~\ref{minpairthm} shows that not only can a \ce
Turing complete set be half of a minimal pair in the tt-degrees, but the
other half of the minimal pair may also be a \ce Turing complete set.

\begin{qu}
Is there a truth-table minimal pair of bT-complete \ce sets?
\end{qu}

If this question has a negative solution, then Theorem~\ref{theorem1} would
follow, since sets of random strings are always bT-complete.

\begin{qu}
Which Turing degrees contain minimal pairs of (c.e.)\ tt-degrees?
\end{qu}

Jockusch~\cite{Jockusch} showed that the hyperimmune-free degrees coincide
with the Turing degrees that contain a single tt-degree; therefore, such
degrees cannot contain a minimal pair of tt-degrees. Jockusch also showed
that if a Turing degree contains more than one tt-degree, it contains an
infinite chain of tt-degrees. It is not known which of the hyperimmune
degrees, apart from~$\0'$, contain a minimal pair of tt-degrees. Not all do:
Kobzev~\cite{Kobzev} proved that there is a noncomputable c.e.\ set $A$ such
that if $B\equiv_T A$, then $A\leq_{\rm tt} B$.\footnote{In particular,
Kobzev showed this for any noncomputable, semicomputable, $\eta$-maximal set
$A$. We thank one of the anonymous referees for pointing us to this result.}
In other words, the tt-degree of $A$ is least among all the tt-degrees in the
Turing degree of $A$, so the Turing degree of $A$ does not contain a minimal
pair of tt-degrees. The $A$ that Kobzev constructed actually has minimal
tt-degree, hence must be low$_2$~\cite{DS}.

\section{\texorpdfstring{No noncomputable set is tt-reducible to every
$R_{K_U}$}{No noncomputable set is tt-reducible to every U-random set}}
\label{sec_X}

We have seen in Theorem~\ref{thmfinite} that given a finite collection of
sets of random strings $\{R_{K_{U_1}},\ldots R_{K_{U_n}}\}$, there is a
noncomputable \ce set tt-reducible to each $R_{K_{U_i}}$.  It is natural to
ask if there is in fact a noncomputable (and perhaps also c.e.)\ set
tt-reducible to every~$R_{K_U}$.  We show that there is no such set.

\begin{thm}\label{knight-bishop}
Given any noncomputable set~$X$, there is a universal prefix-free
machine~${U}$ such that~$X$ is not truth-table reducible to~$R_{K_{U}}$; that
is, there is no common noncomputable information tt-below every~$R_{K_{U}}$.
\end{thm}

Note that this theorem is in contrast to the non-prefix-free case, since
every~$R_{C_U}$ is tt-complete.

\begin{proof}
We begin by giving a sketch of the construction. We will construct three
different prefix-free universal machines ${U}_0,{U}_1,{U}_2$, and guarantee
that they cannot all tt-compute~$X$. For convenience of notation, we denote
the corresponding $R_{K_{U}}$'s by~$R_0$,~$R_1$, and~$R_2$. At the moment, we
do not know whether this non-uniformity is necessary in the proof.

Since every~$R_{K_{U}}$ is~$\Delta^0_2$, we need only consider
$\Delta^0_2$-sets~$X$.  Let $\{\Psi_i\}_{i\in \omega}$ be a listing of
partial tt-reductions.  We will meet the following requirements for all~$i$:
\[
\mathcal{R}_i: \neg(\Psi_i^{R_0}=\Psi_i^{R_1}=\Psi_i^{R_2}=X)
\]
By Posner's trick, this is enough to show that~$X$ is not tt-reducible to all
three sets, as if it were, we could build a single tt-reduction~$\Psi_i$ such
that $\Psi_i^{R_j}=X$ for each $j \le 2$.  To satisfy
requirement~$\mathcal{R}_i$, either~$\Psi_i$ will not be a total
tt-reduction, or we will force one of the following to hold:

\begin{enumerate}[label=(\roman*)]
\item $\Psi_i^{R_j}(x)\neq \Psi_i^{R_k}(x)$ for some $x\in \omega$ and
some $j,k \le 2$, or
\item $\Psi_i^{R_j}\neq X$ for some $j \le 2$.
\end{enumerate}
The way we will achieve this is to build the machines in such a way that if
condition~(i) fails, then the set $\Psi_i^{R_j}$ must be computable, so it
cannot be~$X$.

In order to make these machines universal, we fix a universal prefix-free
machine~${V}$ and simply require that ${U}_j(000*\sigma)={V}(\sigma)$ for
each $j \le 2$. We consider this coding requirement as our \emph{opponent}
controlling~$1/8$ of the total measure, and the diagonalization requirement
as ``we'', the other player controlling the remaining~$7/8$ of the game board
(machines we build). Here the number of~$0$'s is picked so that we have
some amount of space bigger than our opponent as needed in the verification
process.

\subsection{A single requirement \texorpdfstring{$\mathcal{R}_0$}%
{R0}}

We first consider how to satisfy only one requirement~$\mathcal{R}_0$.


\subsubsection{One-bit game}\label{onebit}
We begin by considering only one bit,~0; that is, we are looking only at the
first bit of the first tt-reduction.  We wait until the truth table of
$\Psi_0(0)$ is defined.  Since~$\Psi_0$ may not be a total tt-reduction, this
may never occur.   Before this happens, we do nothing for this
requirement~$\mathcal{R}_0$. Once we have the truth table for~$\Psi_0(0)$, we
can attempt to satisfy this requirement.

We modify the games used in Muchnik's proof that there is a universal
prefix-free machine~${U}$ such that~$R_{K_{U}}$ is not tt-complete.  For the
moment, we define the game ``$G(\epsilon,\delta)$ on~$R_j$'' as follows. We
imagine that the game board is the truth table of $\Psi_0(0)$ and that our
starting position on the game board is the current state of~$R_j$.  The game
$G(\epsilon,\delta)$ is the game where the opponent (the coding requirement)
has~$\epsilon$ measure to use, and we have~$\delta$ measure to use to
enumerate strings (to change~$R_j$).  We are building KC sets as defined in
Theorem~\ref{theorem1} to construct the~$ U_j$'s so that they will indeed be
prefix-free machines, so we must keep the weight of the sets below~1.
Since~$R_j$ is the set of strings that are random with respect to~$U_j$, we
change~$R_j$ by compressing strings.  Each move consists of a player (the
opponent or us) compressing any number of strings, which may change bits
of~$R_j$ from~1's to~0's.

When $\epsilon=\delta$, i.e., when the game is symmetric, we always have a
winning strategy for forcing $\Psi_0^{R_j}(0)$ to be either~$0$ or~$1$ for
each~$R_j$.  We call the value being forced the \emph{value of game
$G(\epsilon,\epsilon)$ on~$R_j$}.  Note that we can computably determine the
value of the game since the game is finite and has only finitely many
sequences of play.

For now fix a small~$\epsilon_0$ (we call this~$\epsilon_0$ the
\emph{starting measure} of the requirement~$\mathcal{R}_0$). If for some $j,k
\le 2$, the values of the games $G(\epsilon_0,\epsilon_0)$ played on~$R_j$
and~$R_k$ are different, i.e., we have strategies that can force
$\Psi_0^{R_j}(0)\neq\Psi_0^{R_k}(0)$, then for the least such pair $j,k$, we
use the strategies for both and play the games with the opponent.

There are two possible outcomes of this dual game. First, if the opponent
never uses more than~$\epsilon_0$ measure (i.e., he does not cheat in the
game), then we satisfy the requirement~$\mathcal{R}_0$ in finitely many
stages by forcing a disagreement between $\Psi_0^{R_j}$ and $\Psi_0^{R_k}$.
If the opponent uses more than~$\epsilon_0$ measure in the play, then we
simply reset the game.  Note that in this situation the opponent uses more
measure than we do.  In the end, he can only cheat by using over~$\epsilon_0$
measure finitely often, since his total measure is bounded by~$1/8$.

In the case where we cannot find such a pair $j,k$, we know that for the
games $G(\epsilon_0,\epsilon_0)$ played on each set $R_0,R_1,R_2$, the values
have to be the same. Now reduce the measure and consider the games
$G(\epsilon_0/2,\epsilon_0/2)$ on each set and compare the values to the
values given in the original games.

We first deal with the scenario when there exist $j,k$ such that
$G(\epsilon_0,\epsilon_0)$ on~$R_j$ and $G(\epsilon_0/2,\epsilon_0/2)$
on~$R_k$ have different values. In this case, we play both games at the same
time, forcing the values to be different. If the opponent does not cheat,
then we have a permanent win. If the opponent cheats, we will do a modified
game analysis (see \S~\ref{bishop}) to resettle agreement on the games.

The remaining case is that these two levels of games,~$\epsilon_0$ and
~$\epsilon_0/2$, on all three sets, all have the same value. In this case, we
continue to look at the next level~$\epsilon_0/4$, then~$\epsilon_0/8$, and
so on.  We call this sequence of games the \emph{stack of games} for the
first bit.  If we find a game, say
$G(\epsilon_0/{2^{n+1}},\epsilon_0/{2^{n+1}})$, on~$R_j$ that has a different
value from all previous games, then we play that game simultaneously with the
game $G(\epsilon_0/{2^n},\epsilon_0/{2^n})$ on~$R_k$ for the least $k\neq j$.
It will be important to always choose the second game from the previous level
and not from another earlier level, so that if the opponent cheats in the
game $G(\epsilon_0/{2^{n+1}},\epsilon_0/{2^{n+1}})$, we will know that~$R_k$
has only used at most twice the measure that the opponent used.

The tt-reduction has been fixed, so eventually we reach a small enough
measure so that the game is actually the $0$-game, i.e., no one can enumerate
anything to change the tt-reduction, or the ``game board''. In this case, the
game is already determined by the current value of the tt-reduction, and in
such a case, we check if the current $\Psi_0^{R_j}(0)=X_s(0)$.  Note that
since we were unable to force a disagreement, this value will be the same for
each $j \le 2$.  If $\Psi_0^{R_j}(0)\neq X_s(0)$, then we stop considering
this requirement~$\mathcal{R}_0$ since the requirement seems to be satisfied.
If~$X(0)$ changes value later we will continue the construction. If the two
values agree, then we move on to consider $\Psi_0\res 2$, i.e., the first two
bits of~$\Psi_0$ (see \S~\ref{multiple bits}).  We will show that this
process cannot continue forever, else~$X$ would be computable, so we will
eventually satisfy requirement $\mathcal{R}_0$.

\subsubsection{``Knight and Bishop'' strategy}\label{bishop}

Now we discuss how to handle the scenario when the opponent cheats in an
intermediate level game, e.g., $G(\delta,\delta)$ on~$R_0$ and
$G(\delta/2,\delta/2)$ on~$R_1$ (other cases are analogous). Note that
whenever his cheat amount is greater than~$\epsilon_0$, we can always reset
the whole stack of games, as we know we will only have to do this finitely
often.

Now if the measure used by the opponent does not exceed~$\epsilon_0$ but
exceeds the amount he is allowed to use in either game, i.e., he cheats, then
we reset the games on~$R_0$ and~$R_1$, and consider brand-new games
$G(\epsilon_0,\epsilon_0)$ on these two sets (with the current game boards).
On~$R_2$, note that the opponent has used the same amount of measure, as his
actions on these three game boards are identical, but we haven't done
anything.  Consider the game $G(\epsilon_0-\lambda,\epsilon_0)$ on~$R_2$,
where~$\lambda$ is the amount the opponent has already used.

This modified game is not symmetric, but it is easy to see that we can force
the same value here as the value we could force for
$G(\epsilon_0,\epsilon_0)$ on~$R_2$ when we started playing
$G(\delta,\delta)$ on~$R_0$ and $G(\delta/2,\delta/2)$ on~$R_1$. The reason
is that we have not yet made any move on~$R_2$ since, and so we may regard
all of the opponent's actions since as the first move of his play, and we can
simply use the same winning strategy to force the same value. Note that our
winning strategy did not depend on the turn order of the game, as each player
is only capable of changing~1's to~0's in~$R_2$, so turn order is not
important and we may allow that the opponent plays first. In the construction
in \S\ref{const}, the opponent is always given the opportunity to play first.

Now if the new games $G(\epsilon_0,\epsilon_0)$ on~$R_0$ or~$R_1$ have
different values from the modified game on~$R_2$, then we can play the new
game on~$R_0$ (or~$R_1$) and the modified game on~$R_2$ to force a
difference. If the opponent again cheats, then together with the amount he
already used before, he must have exceeded his allowed measure~$\epsilon_0$,
and so we can reset the whole stack of games.

If these games all have the same value, we have reset the agreement on~$R_0$
and~$R_1$ for the first bit, and the new value being forced is the same as
the old value (before cheating). Now we consider the game
$G(\epsilon_0,\epsilon_0)$ on~$R_2$, which could have a new value as the game
board has changed since we previously considered this game. This goes back to
the original set-up of symmetric games at the~$\epsilon_0$ level, and so we
can continue the construction.  Note that the opponent can cheat only
finitely often because the stack of games is finite, so there is some minimal
measure $\epsilon_0/{2^n}$ that the opponent must have used in order to
cheat, and the opponent's total measure is bounded by~$1/8$.

In the above discussion, we can think of~$R_0$ and~$R_1$ as knights who have
gone off to fight a battle.  Their opponent has cheated and they return home.
The bishop,~$R_2$, is waiting for them and restores their faith when they
return.  If the three new games $G(\epsilon_0,\epsilon_0)$ all force the same
value, it will be the same value as before.  We will use this in the
verification to show that if there is no disagreement between the three
tt-reductions, then the set they are computing must be computable and so it
cannot be~$X$.  The idea is that if by stage~$s_1$, the opponent has stopped
exceeding the~$\epsilon_0$ limit, then any time after stage~$s_1$ that the
values of the games $G(\epsilon_0, \epsilon_0)$ agree, this value will always
be the same.  To see this more clearly, we first must discuss the
multiple-bit game.

\subsubsection{Multiple bits}\label{multiple bits}

For the one-bit game as above, once we have a stack of games
from~$\epsilon_0$ to~$0$ (remember that a sufficiently small game is already
the $0$-game, where neither player can change the game board), then we check
whether the value $\Psi_0^{R_j}(0)$ we have forced agrees with the current
$X_s(0)$. If not, then we stop considering the requirement~$\mathcal{R}_0$;
if so, we continue to look at the $2$-bit game and similarly build such a
stack of games. So now, by induction, let us consider an $n$-bit game, i.e.,
we consider the first~$n$ many bits of~$\Psi_0$.
An $n$-bit game $G(\epsilon,\delta)$ on~$R_j$ is defined similarly to the
one-bit game.  The game board is now the set of all truth tables for the
first~$n$ bits of~$\Psi_0$, which may be thought of as one large truth table.
The starting position is again the current state of~$R_j$.

Again we wait until the truth tables for each of the $n$ bits have been
determined. Now the situation is slightly more complicated. Consider the
first game board~$R_0$. Given a set $S\subseteq 2^n$, when we play the game
$G(\epsilon_0,\epsilon_0)$ on~$R_0$, by symmetry we have a winning strategy
for either~$S$ or its complement~$\overline{S}$; that is, we can force the
sequence of values of the first~$n$ bits of $\Psi_0^{R_0}$ to be in
either~$S$ or its complement. If we have a winning strategy for~$S$, then we
call~$S$ a {\em winning set}. The collection of all such winning sets gives
us a collection of subsets of~$2^n$. If two games~$R_j$ and~$R_k$ do not have
the same collection of winning sets, then there is an~$S$ such that we have
winning strategies for~$S$ on~$R_j$ and~$\overline{S}$ on~$R_k$, both for the
game $G(\epsilon_0,\epsilon_0)$. Then we can simply start using the
strategies to play the game with the opponent, and the requirement is
satisfied unless the opponent cheats by using more than~$\epsilon_0$ amount
of measure, in which case we reset the whole game board. Thus, if we cannot
start playing a game to cause a disagreement between $\Psi_0^{R_j}$ and
$\Psi_0^{R_k}$, then we may assume that all three sets~$R_0$,~$R_1$,
and~$R_2$ have the same collection of winning sets.

We may also assume that the collection of winning sets forms an ultrafilter.
We have already mentioned that for any $S\subset 2^n$, either~$S$ or its
complement is a winning set.  It is also easy to see that if $S\subset T$
and~$S$ is a winning set, then so is~$T$.  It is left to show that the
intersection of two winning sets~$S$ and~$T$ is also a winning set.  For a
contradiction let us assume that~$S$ and~$T$ are winning sets while $S\cap T$
is not. Then we know that $\overline{S\cap T}$ is a winning set.  But then we
can simultaneously play three games using the winning strategies of~$S$
on~$R_0$,~$T$ on~$R_1$ and $\overline{S\cap T}$ on~$R_2$, causing a
disagreement between $\Psi_0^{R_j}$ and $\Psi_0^{R_k}$ for some~$j$ and~$k$
as long as the opponent does not cheat, in which case we reset the whole game
board.  Thus, if we cannot start playing a game to force a difference in this
way, then the collection of winning sets must be closed under intersection.

In a finite Boolean algebra such as the collection of subsets of~$2^n$, every
ultrafilter is principal, so the collection of winning sets is generated by a
single $\sigma\in 2^n$. Thus, this singleton set $\{\sigma\}$ is itself a
winning set for the games $G(\epsilon_0,\epsilon_0)$ on these three sets.

Note that this $\sigma\in 2^n$ is compatible with the $\tau\in 2^{n-1}$ we
found at the last step when we considered $(n-1)$-bit games.  This is because
the set of both extensions of~$\tau$ of length~$n$ forms a winning set,
so~$\sigma$ must be an extension of~$\tau$.

Now the construction proceeds in a similar way as in the one-bit game. We
consider the next level $G(\epsilon_0/2,\epsilon_0/2)$. If any of the three
games has the complement of $\{\sigma\}$ as a winning set, then we can play
the corresponding games to force a difference; for example,
$G(\epsilon_0,\epsilon_0)$ on~$R_0$ using the strategy for $\{\sigma\}$ and
$G(\epsilon_0/2,\epsilon_0/2)$ on~$R_1$ using the strategy for
$\overline{\{\sigma\}}$. If the opponent cheats, then we will handle it in
the same way as in \S~\ref{bishop}, using the third set to resettle
agreement. We will see in the verification that if we have reached a
stage~$s_1$ by which the opponent has stopped cheating by
exceeding~$\epsilon_0$, then if there is an agreement between all three
tt-reductions, the first~$n$ bits of the set they compute can be determined
by stage~$s_1$. Since this does not depend on~$n$, they would compute~$X$,
which we assumed to be noncomputable.

This finishes the induction step and the analysis of a single
requirement~$\mathcal{R}_0$.

\subsection{Multiple requirements}

To handle multiple~$\mathcal{R}_i$-requirements, we follow one simple rule:
Whenever a higher-priority game acts, then we reset all lower-priority games
and reset their starting measure~$\epsilon_i$ to be a new small number so
that any game playing with that measure will not change the game board for
higher-priority games.  The possible actions of the higher-priority game that
will lead to resetting the lower-priority games include convergence of a
tt-reduction so that a new relevant truth table is defined,
examining new games, and making a move in a game (as defined formally in
Remark~\ref{rem-acted}).

\subsection{Construction}\label{const}
Let $\{X_s\}_{s\in \omega}$ be a computable approximation of the
$\Delta^0_2$-set~$X$.

We construct~$U_j$ by building KC sets~$A_j$ for each $j \le 2$.  Given a
universal prefix-free machine~$V$ and its corresponding \ce KC set~$A_V$, the
opponent enumerates $\langle d+3,\tau\rangle$ into each~$A_j$ whenever
$\langle d, \tau \rangle$ is enumerated into~$A_V$.  In our construction,
whenever we want to enumerate additional elements into~$A_j$ for the purpose
of our games, it will be to make a particular string~$\tau$ be nonrandom with
respect to~$U_j$; that is, the only move we can make is to enumerate some
$\langle d,\tau\rangle$ into~$A_j$ with $d<|\tau|$ so that $R_j(\tau)$
changes from a~1 to a~0.  Therefore, whenever we (and not the opponent)
enumerate elements into~$A_j$, the element may be assumed to be of the form
$\langle |\tau|-1, \tau\rangle$; so when we say that we are enumerating a
string~$\tau$ into~$A_j$, we are actually enumerating $\langle |\tau|-1,
\tau\rangle$.

Begin with $A_{j,0}=\es$ for all $j \le 2$, and with all~$\epsilon_i$
undefined.

Stage $s+1$, $s=\langle i, e\rangle$.  We will act for
requirement~$\mathcal{R}_i$ if able.  We call all stages of the form $\langle
i,e\rangle+1$ {\em i-stages}.

First we allow the opponent to make any enumerations into $A_{j,s+1}$ for
$j \le 2$ as elements enter $A_{V,s+1}$.

\medskip
{\emph Case $1$.} Either $e=0$ or for some $m<i$, $\mathcal{R}_m$ has acted
since the last $i$-stage.

We must define a new~$\epsilon_i$.  Reset any previous value of~$\epsilon_i$
and define the new value of~$\epsilon_i$ to be the greatest number of the
form~$2^{-c}$ where $c\geq i+4$ such that $2^{-c}<\epsilon_m/2$ for all
previously defined values of~$\epsilon_m$ for any $m\in \omega$, and such
that no element in any currently defined truth table for~$\Psi_m$ (for $m<i$)
corresponds to a string of length greater than~$c$.  This will ensure that we
do not add too much measure to~$A_j$ and that games for lower-priority
requirements do not alter the game boards of games for higher-priority
requirements.  After defining~$\epsilon_i$, go to the next stage.

\medskip
{\emph Case $2$.} $\mathcal{R}_m$ has not acted for any $m<i$ since the last
$i$-stage,~$\epsilon_i$ is currently defined, and we are not currently
playing any games for~$\mathcal{R}_i$.

Check if the truth table of $\Psi_i(0)$ has converged after~$s$ steps.  If
not, go to the next stage.  If so, we examine the one-bit games.  Make a
stack of games as described in \S\ref{onebit} and check if there are $j,k \le
2$ and corresponding games in the stack so that we can force a disagreement.
If so, we begin to play the appropriate game.  We give the opponent the
opportunity to move first, which is to say that we will not make a move at
this stage.

If all games throughout the one-bit stack agree, then we ask if they agree
with $X_s(0).$  If not, then go to the next stage.  If so, then we move on to
two bits, and so on.  When we get to~$n$ bits, we check if the truth tables
for $\Psi_i\res n$ have been defined after~$s$ steps.  If not, go to the next
stage.  If so, we determine the winning sets for the games $G(\epsilon_i,
\epsilon_i)$ for $R_j$ for each $j \le 2$.

If the collections of winning sets differ on~$R_j$ and~$R_k$ for some
\mbox{$j,k \le 2$}, then we choose an $S\subset 2^n$ such that~$S$ is a
winning set for~$R_j$ and~$\overline{S}$ is a winning set for~$R_k$ and we
begin to play the games using these strategies. As before, we go to the next
stage, allowing the opponent to move first.

If the collections of winning sets are the same for all the~$R_j$'s, we then
check if there are any winning sets~$S$ and~$T$ such that their intersection
is not a winning set.  If so, we begin to play three games, corresponding to
the strategies for~$S$,~$T$, and $\overline{S\cap T}$.  We go to the next
stage, as usual.

In the remaining situation, the collection of winning sets forms an
ultrafilter generated by some~$\sigma$.
%
%
We can examine the stack of games to see if any of the $G(\delta,\delta)$
games have~$\overline{\{\sigma\}}$ as a winning set.  If so, we begin to play
the appropriate games to force a disagreement and move to the next stage.
Otherwise, all games in the stack have~$\{\sigma\}$ as a winning set, so we
ask if $X_s\res n=\sigma$.  If not, the requirement is temporarily satisfied
and we go to the next stage.  If $X_s\res n=\sigma$, we must consider the
$(n+1)$-bit situation.

(Note that Case~2 also encompasses the situation where we are simply waiting
for either a truth table to be defined or for~$X_s$ to change so that it
agrees with the current tt-reduction.  Thus, the steps of checking the stacks
and finding~$\sigma$, for example, may be repeated unnecessarily in this
construction.)

\medskip
{\em Case $3$.} $\mathcal{R}_m$ has not acted for any $m<i$ since stage the
last $i$-stage,~$\epsilon_i$ is currently defined, and we have already begun
playing games.
Check if the opponent has cheated by enumerating more than his allowed value
in a game.

\smallskip
{\em Case $3$a.} The opponent has not cheated. For each $j \le 2$ such that we
are playing a game on~$R_j$, we follow our designated strategy, which entails
enumerating some set of strings into $A_{j,s+1}$.  We then go to the next
stage, allowing the opponent to play.


\smallskip
{\em Case $3$b.} The opponent has cheated by exceeding $\delta<\epsilon_i$ in
the game $n$-bit game $G(\delta,\delta)$ for $n\geq 1$.  If this happens,
then we were playing two games, on, for example,~$R_0$ and~$R_1$.  We apply
the knight and bishop strategy of \S\ref{bishop}.  We ask if either new
$n$-bit game $G(\epsilon_i,\epsilon_i)$ on~$R_0$ and~$R_1$ has a winning
strategy that could cause a disagreement with the game
$G(\epsilon_i-\lambda,\epsilon_i)$ on~$R_2$, where~$\lambda$ is the amount
used by the opponent since we started the game in which he cheated.  If so,
then we begin to play the appropriate games and move to the next stage.  If
not, then we simply move to the next stage.  (Note that in this situation, if
Case~1 does not apply, then Case~2 will apply and we will once again be
considering the games $G(\epsilon_i,\epsilon_i)$.  On~$R_0$ and~$R_1$, these
games will have the same winning sets through~$n$ bits as they previously
had, because the ``bishop''~$R_2$ has brought them back to their old values.)

\smallskip
{\em Case $3$c.} The opponent has cheated by exceeding~$\epsilon_i$ or by
exceeding $\epsilon_i-\lambda$ in an unbalanced game, as described in Case
3b. Stop playing the games and proceed exactly as in Case~2. The game boards
will have changed since the last time we performed the steps of Case~2.

\begin{rem}\label{rem-acted}
We say the requirement~$\mathcal{R}_i$ has ``acted'' at stage~$s$, thus
causing Case~1 to apply at the next $m$-stages for $m>i$, if any of the
following occur:
\begin{enumerate}[label=(\roman*)]
\item Case 1 applies,
\item We examine games corresponding to a previously unexamined
    truth table (either by a truth table becoming defined or by moving to
    an additional bit),
\item We begin a new game, or
\item Case 3 applies and either we or the opponent makes a nonempty
    move in a game.
\end{enumerate}
\end{rem}

\subsection{Verification}

\begin{lem}
Every requirement~$\mathcal{R}_i$ eventually stops acting and is satisfied.
\end{lem}

\begin{proof}
We follow a standard finite-injury argument. Induct on~$i$. Assume that for
all $m<i$, $\mathcal{R}_m$ has stopped acting by stage~$s_0$. Thus, the
starting measure~$\epsilon_i$ also settles down. Since the opponent cannot
exceed measure~$1/8$, he will only cheat by exceeding~$\epsilon_i$ finitely
many times. Let $s_1>s_0$ be a stage after which the opponent never uses more
than~$\epsilon_i$ measure that affects the $\mathcal{R}_i$-games.

Starting from stage~$s_1$ and the one-bit game, we can always assume that the
tt-reduction converges to define a truth table, since otherwise we have an
automatic satisfaction and the requirement stops acting when it is waiting
for the tt-reduction to converge.

In addition, we can assume that starting from stage~$s_1$, we never start
playing any~$\epsilon_i$ measure games with the opponent for~$\mathcal{R}_i$,
since otherwise we have a permanent win as the opponent can no longer cheat,
and the requirement~$\mathcal{R}_i$ will eventually stop acting. Furthermore,
any other game started for~$\mathcal{R}_i$ must end in the opponent cheating,
else we would get a permanent win.

Assume for a contradiction that requirement~$\mathcal{R}_i$ is not satisfied.
Then for each $k\in\omega$, we can establish the stacks of games
from~$\epsilon_i$ to the $0$ game for every $k$-bit game.  Note that as
described in \S\ref{bishop}, we may resettle agreement after intermediate
level cheating.  We can see that the set~$X$ is going to be computable since
the~$\sigma$'s as in the construction have to be initial segments of~$X$ in
order for the game to continue forever.  The purpose of the knight and bishop
argument in \S\ref{bishop} was to ensure that intermediate level cheating
could not alter the agreed-upon value of~$\sigma$, so we need only find the
first such~$\sigma$ after stage~$s_1$ to know that~$\sigma$ is an initial
segment of~$X$.  However, we know that~$X$ is in fact not computable, so
there must be some~$k$ such that $X(k)$ is going to be different from
$\sigma(k)$, where~$\sigma$ is agreed upon, hence must be an initial segment
of the tt-reductions from each set~$R_0$,~$R_1$ and~$R_2$.  When this $X(k)$
settles down in its $\Delta^0_2$-approximation and we see that it differs
from $\sigma(k)$, the requirement~$\mathcal{R}_i$ stops acting (possibly
after playing several more games to establish the current stack) and is
permanently satisfied.

In any of the ways in which~$\mathcal{R}_i$ can be satisfied, the requirement
stops acting after some finite stage. Thus, the induction can continue.
\end{proof}

\begin{lem}
In our construction of~$R_j$ for $j \le 2$, we do not exceed the measure we
are allowed to use, namely,~$7/8$.  Thus, the~$U_j$ are each universal
prefix-free machines.
\end{lem}

\begin{proof}
There are three portions of the measure usage. The first is the measure we
use for diagonalization with which we actually have a permanent win in the
end (the ``useful'' measure). The second is the measure we waste when we
reset games when higher-priority requirements act (the ``wasted'' measure);
the third is the measure we lose when the opponent cheats in the games (the
``lost'' measure). It is easy to see that, since each time we reset the
starting measure for a requirement, we pick a new starting
measure~$\epsilon_i$ which can be arbitrarily small, the total amount of the
first and the second portions is easily bounded (by, for example,~$1/4$). In
the construction, our choice of~$\epsilon_i$ led to each~$i$ contributing no
more than $2^{-(i+3)}$ to the measure, so the total amount contributed by
all~$i$ is at most~$1/4$.

For the third portion, we can compare the amount of lost measure to the
amount of measure the opponent uses.  When the opponent cheats, then we use
less than twice the measure the opponent uses.  To see this, note that there
are only two situations where we can use more measure than the opponent uses
in cheating.  One is when there are two games being played simultaneously,
and one is a $G(\delta,\delta)$ game while the other is a
$G(\delta/2,\delta/2)$ game.  The opponent can cheat by exceeding~$\delta/2$,
while we may enumerate up to~$\delta$ measure for~$R_j$.  Thus, we enumerate
less than twice what the opponent enumerates.  The other situation is when
the opponent cheated previously by enumerating~$\lambda$ measure, which led
to us playing a $G(\epsilon_i,\epsilon_i)$ game on~$R_j$ along with a
$G(\epsilon_i-\lambda, \epsilon_i)$ game on~$R_k$.  In the game on~$R_k$,
since we enumerated nothing in the previous game, if the opponent cheats now,
his total measure in the two games will exceed~$\epsilon_i$, while ours
for~$R_k$ will not.  For~$R_j$, we may have enumerated strings into~$A_j$ in
both this game and the previous game.  However, our total for the two games
will not exceed 2$\epsilon_i$, or twice the opponent's measure.  Thus, the
lost measure is bounded by~$1/4$, which is twice the opponent's measure.

Note that we are not double-counting the opponent's moves when accounting for
lost measure. In particular, if $i<j$, then no move in an
$\mathcal{R}_j$-game can affect an $\mathcal{R}_i$-game (by the choice of
$\epsilon_j$). On the other hand, any move in an $\mathcal{R}_i$-game is
counted as an $\mathcal{R}_i$ action, so it resents any current
$\mathcal{R}_j$ game. This means that an opponent's move is only counted in
one game on~$R_k$, for each $k\leq 2$.

Finally, $1/4+1/4=1/2$ bounds the total amount of measure we use in
the construction, which therefore does not exceed the amount we are
allowed to use, namely,~$7/8$.
\end{proof}

This concludes the proof of Theorem~\ref{knight-bishop}.
\end{proof}

\end{document}